\documentclass[11pt]{article}

\usepackage[margin=1in]{geometry}
\usepackage{amsmath,amsthm,amssymb}
\usepackage{environ, tabularx}
\usepackage{tikz-cd,graphicx}
\usepackage{float}
\usepackage{indentfirst}
\usepackage{mathtools}
\usepackage{hyperref}
\usepackage{amsmath}
\usepackage{dirtytalk}
\usepackage{pifont}
\usepackage{soul}

\newcommand{\shell}{\textsf{shell}}
\newcommand{\R}{\mathbb{R}}
\newcommand{\Z}{\mathbb{Z}}

\newcommand{\K}{\mathcal{K}}
\newcommand{\M}{\mathcal{M}}

\newcommand{\NP}{\textsf{NP}}
\newcommand{\APX}{\textsf{APX}}

\makeatletter

\usepackage{aliascnt}
\newtheorem{theorem}{Theorem}[section]

\newaliascnt{lemma}{theorem}
\newtheorem{lemma}[lemma]{Lemma}
\aliascntresetthe{lemma}

\newaliascnt{corollary}{theorem}

\aliascntresetthe{corollary}

\DeclareMathOperator{\im}{im}

\author{Glencora Borradiale \and William Maxwell \and Amir Nayyeri}

\title{Minimum bounded chains and minimum homologous chains in embedded simplicial complexes}
\date{}

\begin{document}

\maketitle

\begin{abstract}
We study two optimization problems on simplicial complexes with homology over $\Z_2$, 
the minimum bounded chain problem: given a $d$-dimensional complex $\K$ embedded in $\R^{d+1}$ and a null-homologous $(d-1)$-cycle $C$ in $\K$, find the minimum $d$-chain with boundary $C$, and 
the minimum homologous chain problem: given a $(d+1)$-manifold $\M$ and a $d$-chain $D$ in $\M$, find the minimum $d$-chain homologous to $D$.
We show strong hardness results for both problems even for small values of $d$; $d = 2$ for the former problem, and $d=1$ for the latter problem.
We show that both problems are \APX-hard, and hard to approximate within any constant factor assuming the unique games conjecture.
On the positive side, we show that both problems are fixed-parameter tractable with respect to the size of the optimal solution.
Moreover, we provide an $O(\sqrt{\log \beta_d})$-approximation algorithm for the minimum bounded chain problem where $\beta_d$ is the $d$th Betti number of $\K$.
Finally, we provide an $O(\sqrt{\log n_{d+1}})$-approximation algorithm for the minimum homologous chain problem where $n_{d+1}$ is the number of $d$-simplices in $\M$.
\end{abstract}

\section{Introduction}

Simplicial complexes are best known as a generalization of graphs, but have more structure than other generalizations such as hypergraphs.  Despite the structure, simplicial complexes are sufficiently expressive to make many algorithmic questions computationally intractable.  For example, the generalization of shortest path that we examine in this work is \NP-hard in 2-dimensional simplicial complexes~\cite{Dunfield2011}.  Since planar graphs (1-dimensional simplicial complexes embeddable in $\R^2$) exhibit structure that is algorithmically useful, resulting in more efficient or more accurate algorithms than for general graphs, we ask whether 2-dimensional simplicial complexes that are embeddable in $\R^3$ (and more generally, $d$-complexes emdeddable in $\R^{d+1}$) also have sufficient structure that can be exploited algorithmically.  To this end, we examine the algebraic generalization of the shortest path problem in graphs to simplicial complexes of higher dimension.  This restriction via embedding in Euclidean space would still result in a useful algorithmic tool, given the connection of embedded simplicial complexes to meshes arising from physical systems.

Formally we study the {\em minimum bounded chain problem} which is the algebraic generalization of the shortest path problem in graphs~\cite{Kirsanov2004}.  The goal of the minimum bounded chain problem is to find a subcomplex whose boundary is a given input cycle $C$.  More precisely: Given a $d$-dimensional simplicial complex $\K$ and a null-homologous $(d-1)$-dimensional cycle $C \subset \K$, find a minimum-cost $d$-chain $D \subset \K$ whose boundary $\partial D = C$.  The requirement that the cycle be null-homologous is necessary and sufficient for the existence of a solution and we study the problem in the context of $\Z_2$-homology.\footnote{Formal definitions are presented in Section~\ref{sec:prelim}.} In $\Z_2$-homology, a $d$-chain is a subset of $d$-simplices of the simplicial complex.
We see this as a generalization of the shortest path problem in graphs as follows:  Let $\K$ be a one dimensional simplicial complex (i.e.\ a graph).  A pair of vertices in the same connected component, $s$ and $t$, is a null-homologous $0$-chain and the minimum 1-chain whose boundary is $\{s,t\}$ is the shortest $(s,t)$-path.
Grady has written on why this generalization is useful in the context of 3D graphics~\cite{Grady2010}.  

The minimum bounded chain problem is closely related to the minimum homologous chain problem which asks: given a $d$-chain $D$, find a minimum-cost $d$-chain $X$ such that the symmetric difference of $D$ and $X$ form the boundary of a $(d+1)$-chain. Alternatively, $X$ is the minimum-cost $d$-chain that is homologous to $D$.  Dunfield and Hirani~\cite{Dunfield2011} show the minimum bounded and homologous chain problems are equivalent under additional assumptions.  We study the minimum homologous chain problem for $d$-chains in $(d+1)$-manifolds.

\subsection{Our results}
We present approximation and fixed-parameter tractable algorithms for the minimum bounded chain and the minimum homologous chain problem. In this paper we consider both problems in the context of simplicial homology over $\Z_2$.  We denote by $n_d$ the number of $d$-simplices of the $d$-dimensional simplicial complex $\K$.
Two of our results assume the unique games conjecture. For an overview of the unique games conjecture and its impact on computational topology we refer the reader to the work of Growchow and Tucker-Foltz \cite{UGCcomptop}.

\begin{theorem}
\label{thm:MBC_apx}
There exists an $O(\sqrt{\log \beta_d})$-approximation algorithm for the minimum bounded chain problem for a simplicial complex $\K$ embedded in $\R^{d+1}$, with $d$th Betti number $\beta_d$.
\end{theorem}

\begin{theorem}
\label{thm:MBC_fpt}
There exists an $O(15^k\cdot k \cdot n_{d}^3)$ time exact algorithm for the minimum bounded chain problem for simplicial complexes embedded in $\R^{d+1}$, where $k$ is the number of $d$-simplices in the optimal solution.
\end{theorem}

\begin{theorem}
\label{thm:MHC_apx}
There exists an $O(\sqrt{\log n_{d+1}})$-approximation algorithm for the minimum homologous chain problem for $d$-chains in $(d+1)$-manifolds.
\end{theorem}

\begin{theorem}
\label{thm:MHC_fpt}
There exists an $O(15^k\cdot k \cdot n_{d}^3)$ time exact algorithm for the minimum homologous chain problem for $d$-chains in $(d+1)$-manifolds, where $k$ is the size of the optimal solution.
\end{theorem}

The running times for the first two theorems is computed assuming that the dual graph of the complex in $\R^{d+1}$ is available. 
The last two theorems hold, more generally, for weak pseudomanifolds studied by Dey et al. in~\cite{Dey2019}.

On the hardness side, we show that constant factor approximation algorithms for these problems (minimum bounded chain and minimum homologous chain) are unlikely.

\begin{theorem}
\label{thm:MBC_hardness}
The minimum bounded chain problem is 
\begin{enumerate}
\item [(i)] hard to approximate within a $(1+\varepsilon)$ factor for some $\varepsilon > 0$ assuming $\textsf{P} \ne \NP$, and 
\item [(ii)] hard to approximate within any constant factor assuming the unique games conjecture, 
\end{enumerate}
even if $\K$ is a $2$-dimensional simplicial complex embedded in $\R^3$ with input cycle $C$ embedded on the boundary of the unbounded volume in $\R^3 \setminus \K$.
\end{theorem}

\begin{theorem}
\label{thm:MHC_hardness}
The minimum homologous chain problem is 
\begin{enumerate}
\item [(i)] hard to approximate within a $(1+\varepsilon)$ factor for some $\varepsilon > 0$ assuming $\textsf{P} \ne \NP$, and 
\item [(ii)] hard to approximate within any constant factor assuming the unique games conjecture, 
\end{enumerate}
even when the input chain is a $1$-cycle on an orientable $2$-manifold.
\end{theorem}

For the sake of completeness, we also give a more general presentation of the result of Kirsanov and Gortler~\cite{Kirsanov2004}, that minimum bounded chain is polynomial time solvable for a $d$-dimensional simplicial complex $\K$ embedded in $\R^{d+1}$ and input chain $C$ null-homologous on the boundary of the unbounded region in $\R^{d+1} \setminus \K$.  This can be found in Section~\ref{sec:exact}.  This algorithmic result is likely to be the most general possible, given Theorem~\ref{thm:MBC_hardness}.

\subsection{Related Work}

\paragraph{Chain problems over $\Z$ and $\R$}

Research on the minimum bounded chain problem is limited to the case of $\Z$-homology, where linear programming techniques can be employed algorithmically.  Sullivan described the problem as the discretization of the minimal spanning surface problem~\cite{Sullivan1990} with Kirsanov reducing the problem to an instance of minimum cut in the dual graph~\cite{Kirsanov2004}.  Sullivan's work is on the closely related cellular complexes, but under the same restrictions are we study (embedded in $\R^d$) and Kirsanov studies the problem in embedded simplicial complexes.

Likewise, research on minimum homologous chain has largely worked in $\Z$-homology.
Dey, Hirani and Krishnamoorthy formulate the minimum homologous chain problem over $\Z$ as an integer linear program and describe topological conditions for the linear program to be totally unimodular (and so, poly-time solvable)~\cite{Dey2011}.  Of course, integer linear programming approaches do not extend to $\Z_2$-homology.

This linear programming approach was then applied to the minimum bounded chain problem (over $\Z$) by Dunfield and Hirani~\cite{Dunfield2011}.  Moreover, they show the minimum bounded chain problem is $\NP$-complete via a reduction from 1-in-3 SAT. The gadget they use was originally used by Agol, Hass and Thurston to show that the minimal spanning area problem is $\NP$-complete~\cite{Agol2006}.

Linear programming techniques have also been used by Chambers and Vejdemo-Johansson to solve the minimum bounded chain problem in the context of $\R$-homology \cite{minarea}.
In $\R$-homology Carvalho et al provide an algorithm finding a (not necessarily minimum) bounded chain in a manifold by searching the dual graph of the manifold \cite{topdim}.
\paragraph{Chain problems over $\Z_2$}

Special cases of the minimum homologous chain problem have been studied in $\Z_2$ homology.  The homology localization problem is the case when the input chain is a cycle.  
The homology localization problem over $\Z_2$ in surface-embedded graphs is known to be \NP-hard via a reduction from maximum cut by Chambers et al.~\cite{Chambers2009}; our reduction is from the complement problem minimum uncut.
On the algorithmic side, Erickson and Nayyeri provide a $2^{O(g)} n \log n$ time algorithm where $g$ is the genus of the surface~\cite{Erickson2011}.  
Using the idea of annotated simplices, Busaryev et al.\ generalize this algorithm for homology localization of $1$-cycles in simplicial complexes; the algorithm runs in $O(n^\omega) + 2^{O(g)}n^2 \log n$ time where $\omega$ is the exponent of matrix multiplication, and $g$ is the first homology rank of the complex~\cite{Annotate}.

Using a reduction from the nearest codeword problem Chen and Freedman showed that homology localization with coefficients over $\Z_2$ is not only $\NP$-hard, but it cannot be approximated within any constant factor in polynomial time~\cite{Chen2011}. These hardness results hold for a 2-dimensional simplicial complex, but not necessarily for 2-dimensional complexes embedded in $\R^3$.  They also give a polynomial-time algorithm for the special case of $d$-dimensional simplicial complex that is embedded in $\R^d$.  (This is different from our setting of a $d$-dimensional simplicial complex that is embedded in $\R^{d+1}$; however the algorithm also reduces to a minimum cut problem in a dual graph, much like that of Kirsanov and Gortler.)

\paragraph{Algebraic formulations}

The minimum bounded chain problem over $\Z_2$ can be stated as a linear algebra problem, but this has little algorithmic use since the resulting problems are intractable.   The algebraic formulation is to find a vector $x$ of minimum Hamming weight that solves an appropriately defined linear system $Ax=b$.  (It is possible to reduce in the reverse direction, but the resulting complex is not embeddable in general, and so provides no new results.)

In coding theory this algebraic problem is a well studied decoding problem known as maximum likelihood decoding, and it was shown to be $\NP$-hard by Berlekamp, McEliece and van Tilborg~\cite{alghard,Vardy97IntractibilityMinDistCodes}.
Downey, Fellows, Vardy and Whittle show that maximum likelihood decoding is $\textsf{W[1]}$-hard~\cite{Downey99ParamComplexCodTheory}.
Further, Austrin and Khot show that maximum likelihood decoding is hard to approximate within a factor of $2^{{(\log n)}^{1 - \epsilon}}$ under the assumption that $\NP \nsubseteq \textsf{DTIME}(2^{(\log n)^{O(1)}})$~\cite{Austrin11DetReducGapMinDistCode}.
This work was continued by Bhattacharyya, Gadekar, Ghosal and Saket who showed that maximum likelihood decoding is still $\textsf{W[1]}$-hard when the problem is restricted to $O(k \log n) \times O(k \log n)$ sized matrices for some constant $k$~\cite{Bhattacharyya16HardnessOfLearningSparse}.

\bigskip

\paragraph{Paper organization.}
In Section~\ref{sec:prelim}, we give formal definitions for the paper.  In Section~\ref{sec:FPT}, we present our approximation algorithms and fixed-parameter tractable algorithms.
In Section~\ref{sec:hardness}, we present our hardness results.

\section{Preliminaries}\label{sec:prelim}
\paragraph{Simplicial complexes}
Given a set of vertices $V$ we define an \emph{abstract simplicial complex} $\K$ to be a subset of the power set of $V$ such that the following property holds: if $\sigma \in \K$ and $\tau \subset \sigma$ then $\tau \in \K$.
We call any $\sigma \in \K$ a \emph{simplex} and define the dimension of $\sigma$ to be $|\sigma| - 1$ if $|\sigma| - 1 = d$ we call $\sigma$ a $d$-simplex.
Further, we call 0-simplices, 1-simplices, and 2-simplices \emph{vertices}, \emph{edges}, and \emph{triangles}.
We define the dimension of $\K$ to be equal to the largest dimension of any simplex in $\K$.
If $\K$ has dimension $d$ we refer to $\K$ as a $d$-simplicial complex or $d$-complex.
We refer to any subset of a $d$-simplex $\sigma$ as a \emph{face} of $\sigma$. If $\tau$ is a face of $\sigma$ with dimension $d-1$ we refer to $\tau$ as a \emph{facet} of $\sigma$.

\paragraph{Homology}
In this paper we work in simplicial homology with coefficients over the finite field $\Z_2$.
Here we briefly define the concepts from homology that will be used throughout this paper.
We assume familiarity with the basics of algebraic topology, and refer the reader to standard references \cite{Hatcher2002, Munkres84} for the details.

Given a simplicial complex $\K$ we define the $d$th \emph{chain group} of $\K$ to be the free abelian group, with coefficients over $\Z_2$, generated by the $d$-simplices in $\K$.
We denote the chain group as $C_d(\K)$ and note that its elements are expressed as formal sums $\bigoplus \alpha_i \sigma_i$ where $\alpha_i \in \Z_2$ and $\sigma_i \in \K$ is a $d$-simplex.
We call the elements of the chain group \emph{chains} or more specifically $d$-chains.
When working over $\Z_2$ there is a one-to-one correspondence between $d$-chains and sets of $d$-simplices in $\K$.
It follows that adding two $d$-chains over $\Z_2$ is the same thing as taking the symmetric difference of their corresponding sets. Hence, we use the notation $\sigma \oplus \tau$ to denote the sum of two $d$-chains. By abuse of notation we will also use $\oplus$ to denote the symmetric difference of sets, but the context should always be clear.

For a $d$-simplex $\sigma$ we define its \emph{boundary} $\partial \sigma$ to be the sum of the $(d-1)$-simplices contained in $\sigma$.
We extend this operation linearly to obtain the \emph{boundary operator} on chain groups, $\partial_d \colon C_d(\K) \rightarrow C_{d-1}(\K)$. We will often drop the subscript when the context is clear.
Note that the composition $\partial_{d-1} \partial_d$ is always equal to the zero map.
If $\partial \sigma = \tau$ we say that $\sigma$ is bounded by $\tau$.
We call a chain $\sigma$ a \emph{cycle} if $\partial \sigma = 0$.

By $Z_d(\K)$ we denote the $d$th \emph{cycle group} of $\K$. This is subgroup of $C_p(\K)$ generated by the $d$-simplices in $\ker \partial_d$.
Similarly, by $B_d(\K)$ we denote the $d$th \emph{boundary group} of $\K$, which is the subgroup of $C_p(\K)$ generated by the $d$-simplices in $\im \partial_{d+1}$.
Since $\partial_{d+1} \partial_d = 0$ we have that $B_d(\K)$ is a subgroup of $Z_d(\K)$.
We define the $d$th \emph{homology group} of $\K$, denoted $H_d(\K)$, to be the quotient group $Z_d(\K) / B_d(\K)$.
The $d$th Betti number of $\K$, denoted $\beta_d$, is defined to be the dimension of $H_d(\K)$.
We call a $d$-chain $\sigma$ null-homologous if it is a boundary, that is $\sigma \in B_d(\K)$.
Further, we call two $d$-chains $\sigma$ and $\tau$ \emph{homologous} if their difference is a boundary, that is $\sigma \oplus \tau \in B_d(\K)$.

\paragraph{Embeddings and duality}
Given a $d$-complex $\K$ an embedding of $\K$ is a function $f \colon \K \rightarrow \R^{d+1}$ such that $f$ restricted to any simplex in $\K$ is an injection.
Further, for any two simplices $\sigma, \tau \in \K$ we require that $f(\sigma) \cap f(\tau) = f(\sigma \cap \tau)$. That is, the images of two simplices only intersect at their common faces.
The function $f$ is an \emph{embedding} of the abstract simplicial complex $\K$. In this paper we make no distinction between $\K$ and an embedding of $\K$. Hence, we use the notation $\K$ to refer to both and refer to $\K$ as an \emph{embedded simplicial complex}.

The Alexander duality theorem, a higher dimensional analog of the Jordan curve theorem, states that $\R^{d+1} \setminus \K$ is partitioned into $\beta_d + 1$ connected components.
Exactly one of these connected components is unbounded, and we refer to the unbounded component as $V_\infty$.
Using this partition we define the dual graph $\K^*$ of $\K$.
$\K^*$ has one vertex for each connected component of $\R^{d+1} \setminus \K$ with the vertex corresponding to $V_\infty$ denoted by $v_\infty$.
Further, $\K^*$ has one edge for each $d$-simplex in $\K$.  There is an edge between two vertices representing connected components $V_1$ and $V_2$ in $\K$ if there is a $d$-simplex contained in the intersection of the topological closures of $V_1$ and $V_2$.  Note that $\K^*$ can have parallel edges and self-loops.
Since each $d$-simplex can be in the closure of at most two connected components we have a one-to-one correspondence between $d$-simplices in $\K$ and edges in $\K^*$.
If $S$ is a set of $d$-simplices in $\K$ we denote their corresponding edges in $\K^*$ by $S^*$.
Similar to planar graphs, there is a duality between $d$-cycles in $\K$ and edge cuts in $\K^*$.
There exists a one-to-one correspondence between $d$-cycles in $\K$ and minimal edge cuts in $\K^*$.
We refer to this correspondence as cycle/cut duality, and it will play a central role in many of our proofs.

By $\shell(\K)$ we denote the \emph{outer shell} of $\K$. This is defined to be the subcomplex of $\K$ consisting of all $d$-simplices whose corresponding edges in $\K^*$ are incident to $v_\infty$.
Equivalently, it is also the subcomplex of $\K$ consisting of all $d$-simplices contained in the boundary of $V_\infty$.

We endow the embedding of a simplicial complex $\K$ with the subspace topology inherited from $\R^{d+1}$.
We call $\K$ a $d$-dimensional \emph{manifold} if every point in its embedding is contained in a neighborhood homeomorphic to $\R^d$.
If every point in the embedding of $\K$ is contained in a neighborhood homeomorphic to either $\R^d$ or the $d$-dimensional half-space we call $\K$ a \emph{manifold with boundary}.

\paragraph{Graph cuts}
Let $G = (V, E)$ be a graph. For any two subsets $V_1, V_2 \subset V$ a $(V_1, V_2)$-cut is a set of edges $E'$ such that the graph $G' = (V, E \setminus E')$ contains no path from $V_1$ to $V_2$.
Often we will consider $(S, \overline{S})$-cuts for some $S \subset V$ where $\overline{S}$ denotes the complement of $S$ in $V$.
By $E_S$ we refer to the edge set corresponding to all edges that have one endpoint in $S$ and the other in $\overline{S}$, which is the minimum $(S, \overline{S})$-cut.
We extend this notation to vertices. For any two vertices $s, t \in V$ an $(s, t)$-cut refers to a set of edges whose removal disconnects $s$ from $t$.

\paragraph{The minimum bounded/homologous chain problems} 
Now we give the formal statement of the minimum bounded chain problem.
Given a $d$-dimensional simplicial complex $\K$ and a $(d-1)$-cycle $C$ contained in $\K$ the \emph{minimum bounded chain} problem $(\K, C)$ asks to find a $d$-chain $X$ with $\partial X = C$ such that the cost of $X$ is minimized.
The cost of $X$ is given by its $\ell_1$ norm $\|X\|_1$. Here we are treating $X$ as an $n$-dimensional indicator vector where $n$ is the number of $d$-simplices in $\K$.
The simplicial complex $\K$ may be weighted by assigning a real number to each $d$-simplex in $\K$.
In this case the cost of $X$ is given by $\langle W, X \rangle$, where $W$ is a vector assigning weights to the $d$-simplices of $\K$.

Now let $D$ be a $d$-chain, which may or may not be a cycle.
The minimum homologous chain problem asks to find a minimum $d$-chain $X$ such that $X = D \oplus \partial V$ for some $(d+1)$-chain $V$,
equivalently, the minimum $d$-chain $X$ such that $D \oplus X$ is null-homologous. The cost of $X$ as well as the weighted problem are defined the same as in the previous paragraph.

In this paper, we study the minimum bounded chain problem for complexes embedded in $\R^{d+1}$, and the minimum homologous chain problem for $d$-chains in $(d+1)$-manifolds.


\section{Approximation algorithm and fixed-parameter tractability}\label{sec:FPT}
In this section, we describe approximation algorithms and parameterized algorithms for both minimum bounded chain and minimum homologous chain problems.
Our algorithms work with the dual graph of the input space.  
In order to simplify our presentation we assume that the dual graph of the input complex contains no loops. The following lemma shows that we can make this assumption without any loss of generality. The proof is in the appendix.

\begin{lemma}\label{preprocess}
In polynomial time we can preprocess an instance of the minimum bounded chain problem $(\K, C)$ into a new instance $(\K', C')$ such that (i) $(\K')^*$ contains no loops and (ii) an $\alpha$-approximation algorithm for $(\K', C')$ implies an $\alpha$-approximation algorithm for $(\K', C')$.
\end{lemma}
\begin{proof}
Let $F$ be the set of $d$-simplices corresponding to the loops in $\K^*$.
The cycle/cut duality implies that no $d$-simplex $f\in F$ can be on a $d$-cycle; as a loop cannot be on any cut. 
Therefore, for any $X, Y$ with boundary $C$, $f$ is either on both of them, or on none of them.  Thus, each $f \in F$ is either on all $d$-chains with boundary $C$, or none of such chains. 
Let $F_{all}\subseteq F$ be the $d$-simplices that are on all $X$ with boundary $C$, and let $F_{none}\subseteq F$ be the $d$-simplices that are on no $X$ with boundary $C$, we have $F_{all}\cup F_{none} = F$.

Now, we compute a feasible solution $Y$ with $\partial Y = C$ by solving the linear system using standard methods \cite{SolveLinear}.
Using $Y$ we can partition $F$ into $F_{all}$ and $F_{none}$: a $d$-simplex $f\in F$ is in $F_{all}$ if it is in $Y$, and in $F_{none}$ otherwise.  
We can remove $F_{none}$ from $\K$ without changing the optimal solution.  
Further, any chain $X$ that bounds $C$ contains all $F_{all}$.  That is, we can write $X = X' \oplus F_{all}$. It follows that 
\[
\partial X = \partial X' \oplus \partial F_{all} = C \Rightarrow \partial X' = C \oplus \partial F_{all} = C'.
\]
Hence, we can find the minimum chain $X'_{opt}$ in $\K' = \K \setminus \partial F_{all}$ that bounds $C'$.  Then, $X_{opt} = X'_{opt} \oplus F_{all}$ is the minimum bounded chain  for $C$ in $\K$.   
Furthermore, any approximation algorithm for $(\K', C')$ implies an approximation algorithm with the same ratio for $(\K, C)$.
To see that, let $X'_{apx}$ be an approximation to $X'_{opt}$ in $(\K', C)$, and let $X_{apx} = X'_{opt} \oplus F_{all}$.  So, we have:
\[
\frac{|X_{apx}|}{|X_{opt}|} = \frac{|X'_{apx} \oplus F_{all}|}{|X'_{opt} \oplus F_{all}|} = \frac{|X'_{apx} \cup F_{all}|}{|X'_{opt} \cup F_{all}|} = \frac{|X'_{apx}| + |F_{all}|}{|X'_{opt}| + |F_{all}|} \leq \frac{|X'_{apx}|}{|X'_{opt}|}.
\]
The second equality holds as $X'_{apx}$ and $X'_{opt}$ are disjoint from $F_{all}$; as they are solutions in $\K'$ that does not contain $F$. 
The last equality holds as $|X'_{apx}|$, $|X'_{apx}|$, and $|F_{all}|$ are non-negative and $|X'_{apx}| \geq |X'_{opt}|$.
\end{proof}

\subsection{Reductions to the minimum cut completion problem}
Given $G = (V,E)$ and $E'\subseteq E$, the \emph{minimum cut completion} problem asks for a cut $(S,\overline{S})$ with edge set $E_S$ that minimizes $|E_S\oplus E'|$.  
First, we show that the minimum cut completion problem generalizes the minimum bounded chain problem.

\begin{lemma}
\label{lem:new_cut_problem}
For any $d$-dimensional instance of the minimum bounded chain problem,  $(\K, C)$,
there exists an instance of the minimum cut completion problem $(G = (V,E), E')$ that can be computed in polynomial time, and a one-to-one correspondence between cuts in $G$ and $d$-chains with boundary $C$ in $\K$.  Moreover, if the cut $(S, \overline{S})$ with edge set $E_S$ in $G$ corresponds to the $d$-chain $Q$ in $\K$ then $|E_S \oplus E'| = |Q|$.  
\end{lemma}
\begin{proof}
Let $F$ be any $d$-chain such that $\partial F = C$, such an $F$ can be computed in polynomial time, by solving the linear system.
In turn, let $G = \K^*$, and $E' = F^*$. 

Now, let $Q$ be any $d$-chain such that $\partial Q = C$.  So, $\partial(Q\oplus F) = 0$. Thus, by cycle/cut duality $Q\oplus F$ partitions $\R^{d+1}$, let $(S,\overline{S})$ be the corresponding dual cut in $\K^*$, and let $E_S$ be the edge set of this cut.  We have $|E_S\oplus E'| = |E_S \oplus F^*| = |Q^*| = |Q|$.

On the other hand, let $(S, \overline{S})$ be a cut in $\K^*$, with edge set $E_S$.  By cycle/cut duality $\partial E_S^* = 0$.  Now, let $Q = E_S^* \oplus F$.  It follows that $\partial Q = C$.  Moreover, we have $|Q| = |E_S^* \oplus F| = |E_S \oplus F^*| = |E_S^* \oplus E'|$.
\end{proof}

Next, we show via a similar argument that the cut completion problem also generalizes the minimum homologous chain problem when the input complex is a \emph{weak pseudomanifold}.
A weak pseudomanifold is a pure $d$-complex such that every $(d-1)$-simplex is a face of at most two $d$-simplices.
Weak pseudomanifolds generalize manifolds and the definition was first introduced by Dey et al. in~\cite{Dey2019}.
Although recognizing $d$-manifolds is undecidable~\cite{ManifoldUnrecognize}, weak pseudomanifolds can be recognized in polynomial time.

\begin{lemma}
\label{lem:new_cut_problem_min_hom_chain}
For any $d$-dimensional instance of the minimum homologous chain problem $(\M, D)$, where $\M$ is a weak pseudomanifold, there exists an instance of the minimum cut completion problem $(G = (V,E), E')$ that can be computed in polynomial time, and a one-to-one correspondence between cuts in $G$ and $d$-chains in $\M$ that are homologous to $D$.  Moreover, if the cut $(S, \overline{S})$ with edge set $E_S$ in $G$ corresponds to the $d$-chain $Q$ in $\K$ then $|E_S \oplus E'| = |Q|$.
\end{lemma}
\begin{proof}
Let $(G=(V, E), E')$ be the instance of the cut completion problem constructed as follows.
For each $(d+1)$-simplex $v^*$ in $\M$, we have a vertex $v$ in $V$.
Moreover, $V$ contains one extra vertex, $v_B$.
For each $d$-simplex $e^*$ in $\M$, if $e^*$ is a face of two $(d+1)$-simplices $v^*$ and $u^*$, we add the edge $(u,v)$ to $E$; we call such an edge a regular edge.
Otherwise, when $e^*$ is only adjacent to one $(d+1)$-simplex $v^*$, we add the edge $(v, v_B)$ to $E$; we call such an edge a boundary edge.
Since $\M$ is a $(d+1)$-weak pseudomanifold, each $d$-simplex is a face of either one or two $(d+1)$-simplices.
Finally, let $E'$ be the edges dual to $D$.

Now, let $Q$ be a $d$-chain in $\M$ that is homologous to $D$. Therefore, $Q\oplus D$ is null-homologous, that is there exists a $d$-chain $S^*$, such that $\partial S^* = Q\oplus D$.  It follows that $(S, \overline{S})$ is a cut in $G$ with cost $|E_S\oplus E'| = |Q|$, as $|E'|$ is dual to $D$ and $E_S$ is dual to $|Q\oplus D|$.

On the other hand, let $(S, \overline{S})$ be a cut in $G$ with edge set $E_S$.  Thus, $E_S^*$ is null-homologous.  So, $Q = E_S^*\oplus D$ is homologous to $D$, and its cost is $|E_S\oplus E'|$.
\end{proof}
\subsection{Algorithms for the minimum cut completion problem}
We show an $O(\sqrt{\log |V|})$-approximation algorithm and a fixed-parameter tractable algorithm for the cut completion problem.  We obtain both of these results via reduction to 
\emph{2CNF Deletion}: given an instance of 2SAT, find the minimum number of clauses to delete to make the instance satisfiable.  
Agarwal et al.~\cite{Agarwal05MinUncut2CNFDeletionETC} show an $O(\sqrt{\log n})$-approximation algorithm for 2CNF Deletion, where $n$ is the number of clauses, and Razgon and O'Sullivan show that the problem is fixed-parameter tractable.

\begin{lemma}[Agarwal et al.\cite{Agarwal05MinUncut2CNFDeletionETC}, Theorem 3.1]
\label{lem:2CNFDeletionApprox}
There is a randomized polynomial-time algorithm for finding an $O(\sqrt{\log n})$-approximation for the minimum disagreement
2CNF Deletion problem.
\end{lemma}

\begin{lemma}[Razgon and O'Sullivan~\cite{Razgon09Almost2SatFpt}, Theorem 7]
\label{lem:2CNFDeletionFPT}
Let $B$ be an instance of 2CNF Deletion problem with $m$ clauses that admits a solution of size $k$.  There is an $O(15^k\cdot k \cdot m^3)$ time exact algorithm for solving $B$.
\end{lemma}

The next lemma shows similar results for the cut completion problem.

\begin{lemma}
\label{lem:cutCompletionAlgs}
For the cut completion problem $(G = (V,E), E')$,
\begin{enumerate}
\item [(i)] there is a randomized polynomial-time $O(\sqrt{\log |V|})$-approximation algorithm, and 
\item [(ii)] there is an $O(15^k\cdot k \cdot |E|^3)$ time exact algorithm, where $k$ is the size of the optimal solution.
\end{enumerate}
\end{lemma}
\begin{proof}
Let $G=(V,E)$, and $E'\subseteq E$.  We show a 2CNF Deletion instance $B_G$ such that for any cut $(S, \overline{S})$ with edge set $E_S$, the number of unsatisfied clauses in $B_G$ is exactly $|E_S\oplus E'|$.  The statement of the lemma will follow from Lemma~\ref{lem:2CNFDeletionApprox} and \ref{lem:2CNFDeletionFPT}.

Let $B_G$ be the instance of the 2CNF Deletion problem defined on $G$ as follows:
\begin{itemize}
\item  For each vertex $v\in V$, we have variable $b(v)$.  
\item For each edge $(u,v)\in E$: 
	\begin{itemize}
		\item if $(u,v) \in E'$, we add $b(u)\vee b(v)$ and $\neg b(u)\vee \neg b(v)$ to $B$, and 
		\item if $(u,v)\notin E'$, we add $b(u)\vee \neg b(v)$ and $\neg b(u) \vee b(v)$ to $B$.
	\end{itemize}
(Note that in both cases, any assignment of $b(u)$ and $b(v)$ satisfies at least one of the clauses.  Again in both cases, assignments exist that satisfy both clauses.)
\end{itemize}

Let $(S, \overline{S})$ be a cut with edge set $E_S$.   Let $b_S$ be the natural boolean vector that corresponds to the cut: $b(v) = [v\in S]$ for all $v\in V$.  
We show that $|E_S\oplus E'|$ is equal to the number of clauses that are not satisfied in $B_G$.  
Specifically, we show (I) for each edge $(u,v)\in E_S\oplus E'$, exactly one of its corresponding clauses is satisfied, and (II) for each edge $(u,v)\notin E_S\oplus E'$ both of its corresponding clauses are satisfied.  

If $(u,v)\in E_S\oplus F$ there are two cases to consider: (I.1) $(u,v)\in E_S$ and $(u,v) \notin E'$, that is $b(u) \neq b(v)$ and the corresponding clauses are $b(u)\vee \neg b(v)$ and $\neg b(u) \vee b(v)$. Exactly one of the clauses is satisfied.  (I.2) $(u,v)\notin E_S$ and $(u,v) \in E'$, that is $b(u) = b(v)$, and the corresponding clauses are $b(u)\vee b(v)$ and $\neg b(u)\vee \neg b(v)$; exactly one of the clauses is satisfied.

If $(u,v)\notin E_S\oplus E'$ there are two cases to consider: (II.1) $(u,v)\in E_S$ and $(u,v) \in E'$, that is $b(u) \neq b(v)$ and the corresponding clauses are $b(u)\vee b(v)$ and $\neg b(u)\vee \neg b(v)$. Both of the clauses are satisfied.  (II.2) $(u,v)\notin E_S$ and $(u,v) \notin E'$, that is $b(u) = b(v)$, and the corresponding clauses are $b(u)\vee \neg b(v)$ and $\neg b(u) \vee b(v)$. Both of the clauses are satisfied.
\end{proof}

\subsection{Wrap up (Proofs of Theorems~\ref{thm:MBC_apx},  \ref{thm:MBC_fpt}, \ref{thm:MHC_apx}, and \ref{thm:MHC_fpt})}
Lemma~\ref{lem:new_cut_problem} and Lemma~\ref{lem:new_cut_problem_min_hom_chain} show that the bounded chain problem and the minimum homologous chain problem are special cases of the cut completion problem, and Lemma~\ref{lem:cutCompletionAlgs} shows that we obtain $O(\sqrt{\log |V|})$-approximation algorithm and $O(15^k\cdot k \cdot |E|^3)$ time exact algorithm for the cut completion problem.  The number of vertices $|V|$ translates to $\beta_d$ for simplicial complexes embedded in $\R^{d+1}$ (Theorem ~\ref{thm:MBC_apx}), and $n_{d+1}$, the number of $(d+1)$-dimensional simplices for $(d+1)$-manifolds (Theorem~\ref{thm:MHC_apx}).
The number of edges $|E|$ translates to $n_d$ in both simplicial complexes embedded in $\R^{d+1}$ (Theorem~\ref{thm:MBC_fpt}) $(d+1)$-manifolds (Theorem~\ref{thm:MHC_fpt}).

\section{Hardness of approximation}\label{sec:hardness}

In this section, we show it is unlikely that either of the minimum bounded chain or minimum homologous chain problems admit constant factor approximation algorithms, even for their low dimensional instances.  Our hardness results follow from reductions from the minimum cut completion problem, defined in the previous section.

\subsection{Minimum bounded chain to minimum cut completion} 
We show that the minimum cut completion problem reduces to a $2$-dimensional instance of the minimum bounded chain problem $(\K, C)$, where $\shell(\K)$ is in fact a manifold and $C$ is a (possibly not connected) cycle on $\shell(\K)$.  Our hardness of approximation result for the minimum bounded chain problem is based on this reduction.

\begin{lemma}
\label{lem:cut_completion_to_bounding_chain}
Let $(G = (V,E), E')$ be any instance of the minimum cut completion problem. 
There exists an instance of the $2$-dimensional minimum bounded chain problem $(\K, C)$ with $C$ on the outer shell of $\K$ that can be computed in polynomial time, and a one-to-one correspondence between cuts in $G$ and $2$-chains with boundary $C$ in $\K$.  Moreover, if the cut $(S, \overline{S})$ with edge set $E_S$ in $G$ corresponds to the $2$-chain $Q$ in $\K$ then 
\[
\frac{|Q|}{\tau} - 1 \leq |E_S \oplus E'| \leq \frac{|Q|}{\tau},
\]
where $\tau = 58m + 2$ and $m$ is the number of edges in $G$.
\end{lemma}
\begin{proof}
Our construction is simple in high-level.  We start from any embedding of $G$ in $\R^3$, and we thicken it to obtain a space, in which each edge corresponds to a tube.  We insert a disk in the middle of each tube; we call these disks \emph{edge disks}.
Then we triangulate all of the $2$-dimensional pieces.
The dual of the complex that we build is almost $G$, except for one extra vertex corresponding to its outer volume, and a set of extra edges, all incident to the extra vertex.
We give our detailed construction below.

We consider the following piecewise linear embedding of $G$ in $\R^3$; let $n$ and $m$ be the number of vertices and edges of $G$, respectively.  
First, map the vertices of $G$ into $\{1, 2, \ldots, n\}$ on the $x$-axis.  Now, consider $m+2$ planes $h_0, h_1, \ldots, h_{m+1}$ all containing the $x$-axis with normals being evenly spaced vectors ranging from $(0, 1, 1)$ to $(0, 1, -1)$.  We use $h_1, \ldots, h_m$ for drawing the edges $G$.  We arbitrarily assign edges of $G$ to these plane, so each plane will contain exactly one edge.  Each edge is drawn on its plane as a three-segment curve; the first and the last segment are orthogonal to $x$-axis and the middle one is parallel.  All edges are drawn in the upper half-space of $\R^3$.  See Figure~\ref{fig:thickend}, left. 

Next, we place an axis parallel cube around each vertex.  The size of the cubes must be so that they do not intersect, fix the width of each cube to be $1/10$. We refer to these cubes as \emph{vertex cubes}
Then, we replace the part of each edge outside the cubes with a cubical tube, called \emph{edge tube}.  
We choose the thickness of these tubes sufficiently small so that they are disjoint.  
We also puncture the cubes so that the union of all vertex cubes and edges tubes form a surface; see Figure~\ref{fig:thickend}, left.  (This surface will have genus $m-n+1$ by Euler's formula, which is the dimension of the cycle space of $G$)

\begin{figure}[!htb]
  \centering
    \includegraphics[height=2.0in]{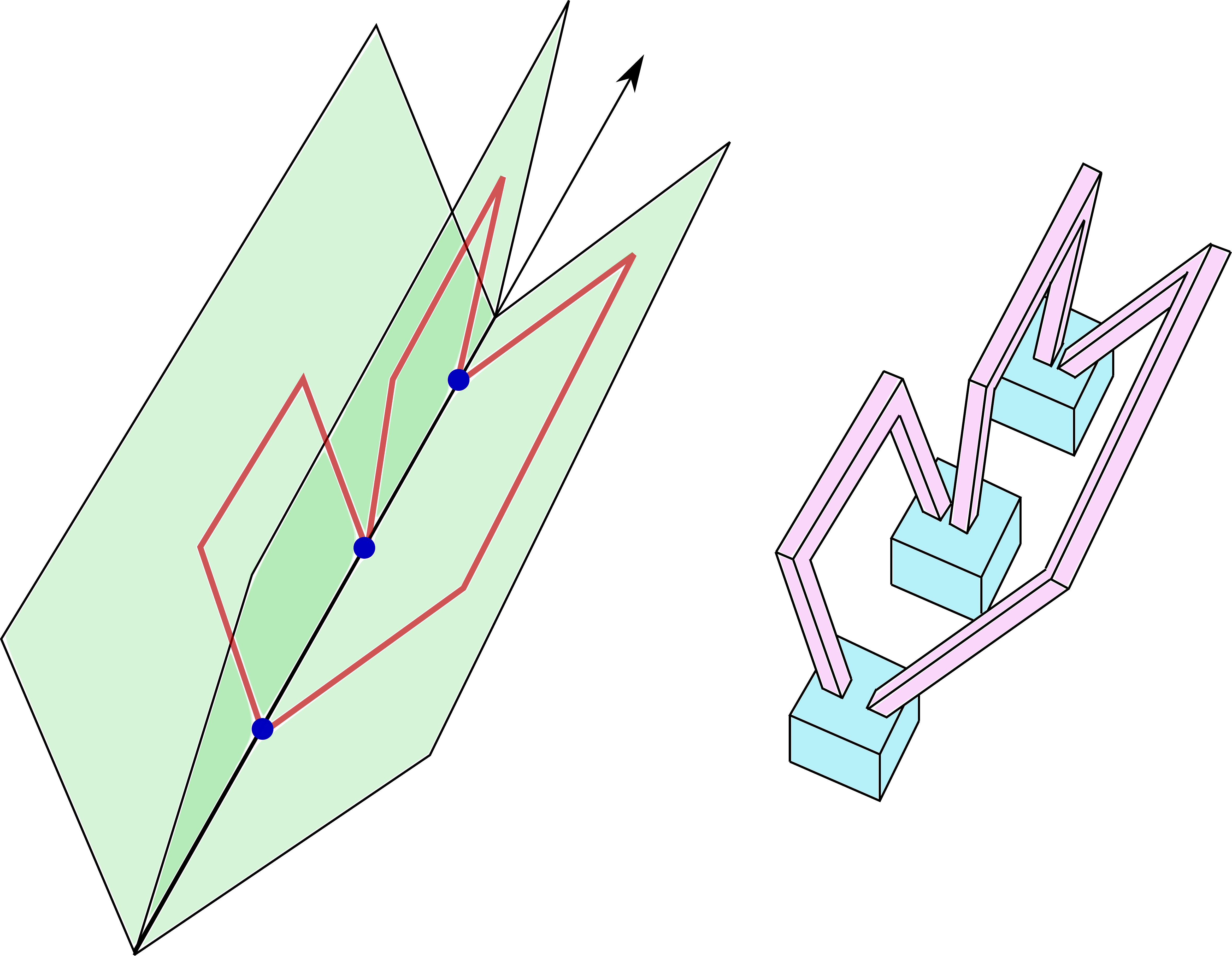} \qquad\qquad\qquad\qquad
    \includegraphics[height=2.0in]{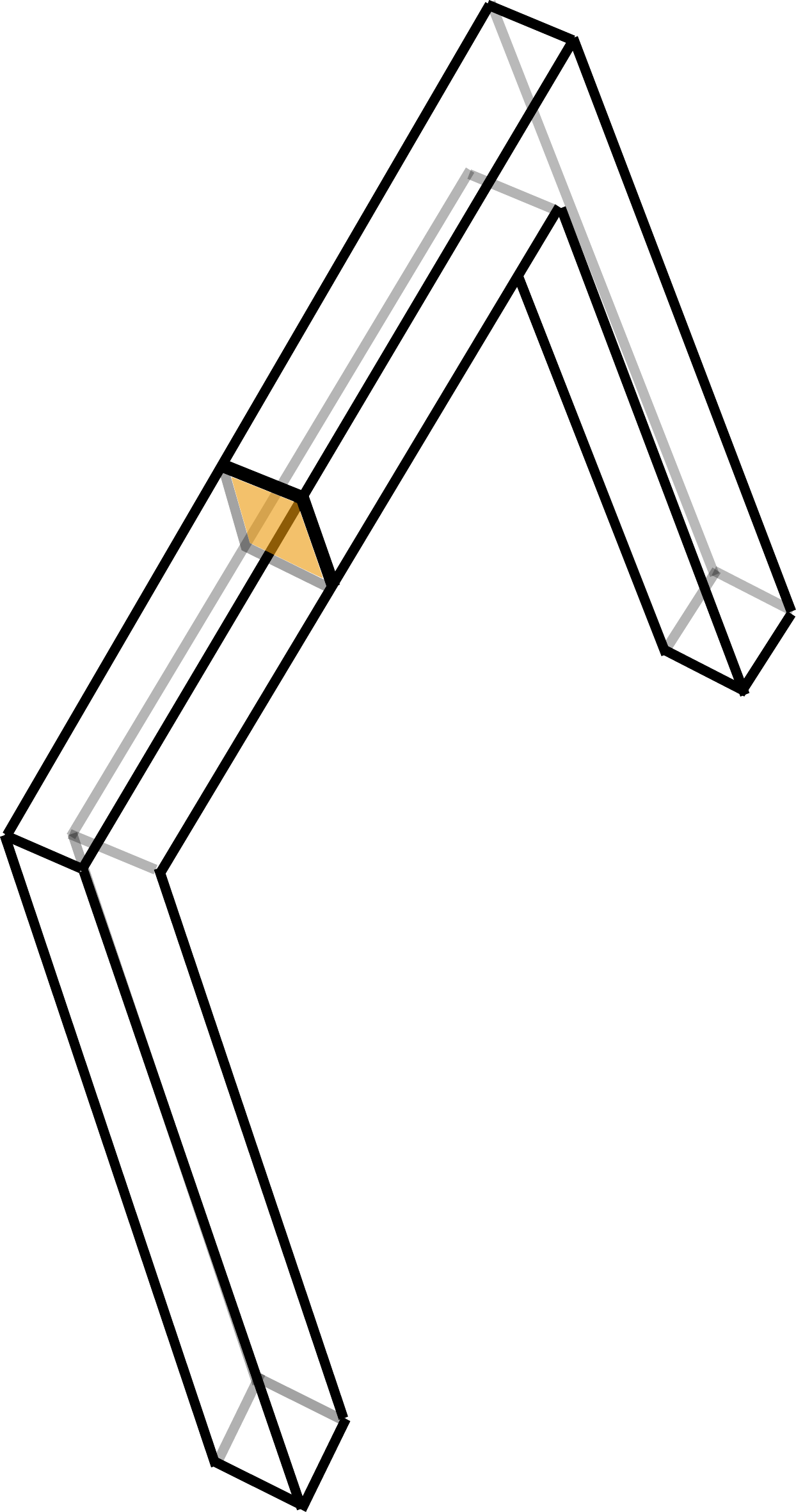}
  \caption{Left: an embedding of $K_3$ in $\R^3$, and the thickened surface composed of blue vertex cubes and pink edge tubes, right: an edge tube subdivided by an edge square.}  
  \label{fig:thickend}
\end{figure}

Next, we subdivide each tube by placing a square in its middle; see Figure~\ref{fig:thickend}, right.  We refer to these squares as \emph{edge squares}.  Edge squares partition the inside of the surface into $n$ volumes.  We observe that each of these volumes contains exactly one vertex of the drawing of $G$, thus, we call them \emph{vertex volumes}.

For our reduction to work, we need that the weight of each $2$-cycle to be dominated by the weight of its edge squares.
To achieve that we finely triangulate each edge square.
For an edge tube, we first subdivide its surface to $16$ quadrangles as shown in Figure~\ref{fig:triangulation}, left.
Then, we obtain a triangulation with $32$ triangles by splitting each quadrangle into two triangles.
For a vertex cube, note that all the punctures are on the top face by our construction.
We split all the other faces by dividing each of them into two triangles.
For the top face, we can obtain a triangulation in polynomial time; this triangulation will have $4\deg(v)+8$ triangles by Euler's formula, where $\deg(v)$ is the degree of the vertex corresponding to the cube.
Therefore, the triangulation of each vertex cube will have $4\deg(v)+18$ triangles, see Figure~\ref{fig:triangulation}, right.
Therefore, there are $\left( \sum_{v \in V}{4\deg(v)+18} \right) + 32m \leq 58m$ triangles that are not part of edge squares.
Finally, we triangulate each edge square into $58m + 2$ triangles so that the cost of one edge square is greater than the sum of all triangles not contained in edge squares. 
This triangulation can be done efficiently by subdividing triangles. The subdivision is performed by inserting a vertex into the interior of the triangle and connecting it with an edge to each vertex on the boundary of the triangle. The result is a new complex, homeomorphic to the original, with two additional triangles.
Overall, our complex $\K$ has $O(m^2)$ triangles. 

\begin{figure}[!htb]
  \centering
    \includegraphics[height=1.5in]{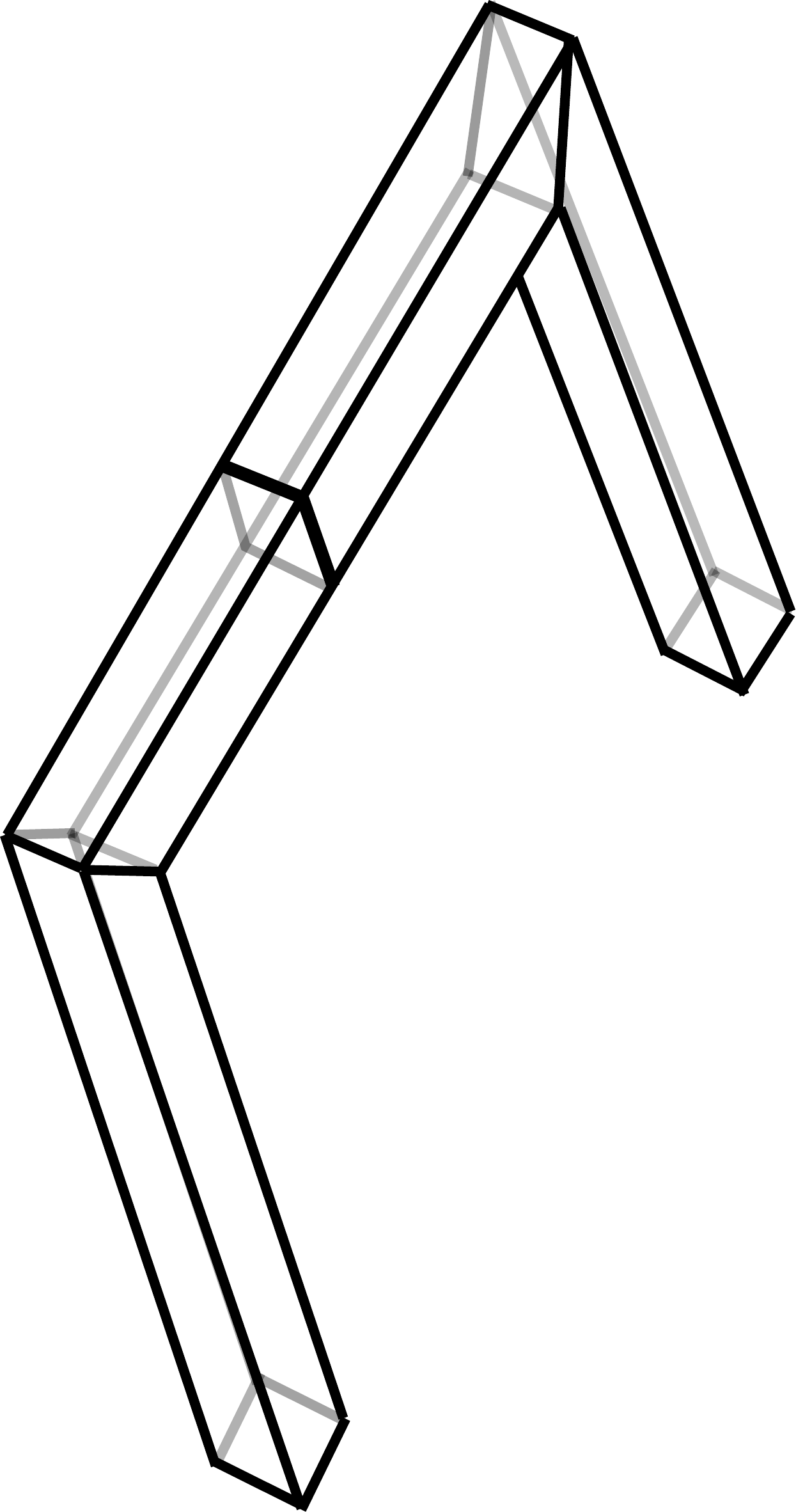} \qquad\qquad\qquad
    \includegraphics[height=1.5in]{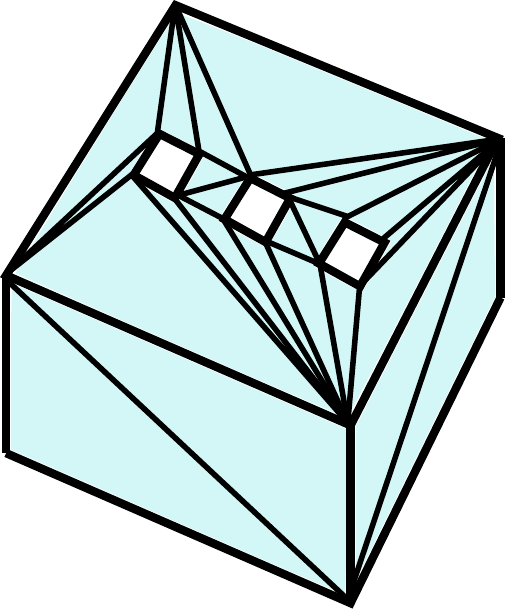}
  \caption{Left: subdividing the surface of an edge-tube to quadrangles, right: triangulating the surface of a vertex cube.}  
  \label{fig:triangulation}
\end{figure}

We are now done with the construction of $\K$.  Let $B$ be the set of all triangles in edge squares that correspond to edges in $E'$. Then, let $C = \partial B$.  
We show an almost cost preserving one-to-one correspondence between cuts in the cut completion problem in $G$ and chains with boundary $C$ in $\K$.

Let $(S, \overline{S})$ be a cut with edge set $E_S$, note that the cost of this cut is $|E_S\oplus E'|$ in the cut completion problem $(G, E')$.
In $\K$, let ${\cal V}_S$ be the symmetric difference of the vertex volumes that correspond to vertices of $S$.
The total weight of ${\cal V}_S$ is between $|E_S| (58m + 2)$ and $|E_S| (58m + 2) + 58m$.
Similarly, the total weight of ${\cal V}_S\oplus B$ is between $|E_S\oplus E'| (58m + 2)$ and $|E_S\oplus E'| ( 58m + 2) + 58m$.
Since we cannot get an exact count on the number of edges in the subgraph induced by $S$ we have a range of values for the weight of $\mathcal{V}_S$ instead of an exact weight.
However, if $E_S$ and $E_{S'}$ are two cuts with $|E_S| < |E_{S'}|$ then the weight of $\mathcal{V}_S$ is strictly less than the weight of $\mathcal{V}_{S'}$ by the construction of the edge squares.

On the other hand, let $Q$ be a $2$-chain with boundary $C$ in $\K$.
As $C$ does not intersect the interior of any edge square, for each edge square either $Q$ contains all of its triangles or none of them.  
Also, $Q\oplus B$ has no boundary, thus its complement $\R^3\setminus (Q\oplus B)$ is disconnected.
The interior of each vertex volume is completely inside one of the connected components of $\R^3\setminus (Q\oplus B)$, as by the construction $Q\oplus B$ must either contain the entire vertex volume or none of it. 
Now, let $S$ be the set of all vertices whose corresponding vertex volumes are in the unbounded connected component of $\R^3\setminus (Q\oplus B)$.  
The edges of the cut $(S, \overline{S})$ correspond to edge squares in $Q_s\oplus B$, where $Q_s$ is the set of edge square triangles of $Q$.
As $B$ is in one-to-one correspondence to $E'$, it follows that the cut completion cost of $(S, \overline{S})$ is $\frac{|Q_s|}{58m + 2}$.
We have $|Q| = |Q_s| + |Q_r|$ where $Q_r$ is the set of triangles in $Q$ not contained in edge squares.
The size of $|Q_s|$ is $58m + 2$ per edge square, and $|Q_r| \leq 58m$ by construction.
It follows that we have our desired inequality, \[\frac{Q}{58m + 2} - 1 \leq |E_S \oplus E'| \leq \frac{Q}{58m + 2}.\]
\end{proof}

The next lemma shows that an approximation algorithm for the minimum bounded chain problem implies an approximation algorithm with almost the same quality for the minimum cut completion problem.

\begin{lemma}
\label{lem:cut_completion_to_bounding_chain_2}
Let $(G = (V,E), E')$ be any instance of the minimum cut completion problem. 
For any $\alpha \geq 1$ and any $\varepsilon > 0$, there exists an instance of the $2$-dimensional minimum bounded chain problem $(\K, C)$ that can be computed in polynomial time,  such that an $\alpha$-approximation algorithm for $(\K, C)$ implies a $((1+\varepsilon)\alpha)$-approximation algorithm for $(G, E')$, and $C$ is on the outer shell of $\K$.
\end{lemma}

\begin{proof}
Let $\varepsilon > 0$.  Given an $\alpha$-approximation algorithm for the minimum bounded chain  problem, we describe an $((1+\varepsilon)\alpha)$-approximation algorithm for the cut completion problem.  Let $G = (V, E)$, and $E'\subseteq E$ be any instance of the cut completion problem, and let $(S_{opt}, \overline{S_{opt}})$ with edge set be an optimal solution for this instance.
Our algorithm considers two cases, based on whether $|E_{S_{opt}} \oplus E'| < 1/\varepsilon$ or not.  It solves the problem under each assumption and outputs the best solution it obtains in the end.

If $|E_{S_{opt}} \oplus E'| < 1/\varepsilon$, then our algorithm finds the optimal solution in $O(n^{1/\varepsilon+O(1)})$ time by considering all subsets of edges $E''$ of size at most $1/\varepsilon$ as candidates for $E_{S_{opt}} \oplus E'$.  From all candidates, we return the minimum $E''$ such that $E'' \oplus E'$ is a cut.  Note this is an exact algorithm, so in this case we find the optimal solution.

Otherwise, if $|E_{S_{opt}} \oplus E'| \geq 1/\varepsilon$, we use the given $\alpha$-approximation algorithm for the minimum bounded chain  problem for a simplicial complex $\K$, and chain $C$ that corresponds to $(G, E')$ by Lemma~\ref{lem:cut_completion_to_bounding_chain}.  Note that $\K$ is an unweighted simplicial complex piecewise linearly embedded in $\R^3$ and $C$ is a cycle in its outer shell.

Let $Q_{opt}$ be the corresponding $2$-chain to $(S_{opt}, \overline{S_{opt}})$ in $\K$.  Thus, $\frac{Q_{opt}}{\tau} - 1 \leq |E_{S_{opt}} \oplus E'| \leq \frac{|Q_{opt}|}{\tau}$.
In addition, let $Q$ be the surface with boundary $C$ that the $\alpha$-approximation algorithm finds, so $|Q| \leq \alpha\cdot|Q_{opt}|$.
Finally, let $(S, \overline{S})$ be the cut corresponding to $Q$ in $G$ via the one-to-one correspondence of Lemma~\ref{lem:cut_completion_to_bounding_chain}. Therefore, $\frac{Q}{\tau} - 1 \leq |E_{S} \oplus E'| \leq \frac{|Q|}{\tau}$.
Putting everything together,
\begin{equation}
\label{eqn:plus_one_bound}
|E_{S} \oplus E'| \leq \frac{|Q|}{\tau} \leq \alpha\cdot\frac{|Q_{opt}|}{\tau} \leq \alpha\cdot\left(|E_{S_{opt}} \oplus E'| + 1\right).
\end{equation}

Since $|E_{S_{opt}} \oplus E'| \geq 1/\varepsilon$, we have: $|E_{S_{opt}} \oplus E'| + 1 \leq (1+\varepsilon)\cdot|E_{S_{opt}} \oplus E'|$.
Therefore, together with (\ref{eqn:plus_one_bound}), we have a $((1+\varepsilon)\alpha)$-approximation algorithm, as desired.
\end{proof}

\subsection{Minimum homologous cycle to minimum cut completion} 
We show a similar reduction from the cut completion problem to the minimum homologous cycle problem for $1$-dimensional cycles on orientable $2$-manifolds.
The minimum homologous cycle problem is the special case of the minimum homologous chain problem when the input chain is required to be a cycle, so showing hardness of approximation for it implies hardness of approximation for the more general minimum homologous chain problem.

\begin{lemma}
\label{lem:cut_completion_to_hom_cycle}
Let $(G = (V,E), E')$ be any instance of the minimum cut completion problem. 
For any $\alpha \geq 1$, there exists an instance of the $1$-dimensional minimum homologous cycle problem $(\M, D)$ that can be computed in polynomial time such that an $\alpha$-approximation for $(\M, D)$ implies an $((1+\epsilon)\alpha)$-approximation for $(G, E')$.
\end{lemma}
\begin{proof}
We construct a 2-manifold $\M$ as in the proof of Lemma \ref{lem:cut_completion_to_bounding_chain}, but we omit the edge squares.
Each edge of $G$ corresponds to a cycle with 4 edges in $\M$; these cycles are the boundaries of the omitted edge squares. We call these cycles \emph{edge rings}.
The connected components of $\M$ after removing the edge rings correspond to the vertices of $G$, we call these connected components \emph{vertex regions}.
We set $D$ to be equal to the set of edge rings corresponding to $E'$.
Intuitively, if $X$ is the minimum cycle homologous to $D$ we do not want $X \oplus D$ to intersect the interior of any vertex region. That is, $X \oplus D$ is a collection of edge rings and corresponds to a cut in $G$.
To achieve this, we subdivide each edge not contained in an edge ring into a long path.
The result is an embedded graph with non-triangular faces, which is not a simplicial complex. To fix this, we triangulate the inside of each non-triangular face such that the shortest path between any two vertices on the face remains the shortest path after the triangulation.
Given any $\alpha$-approximation of the new complex we can obtain a smaller solution using only the edge rings, which corresponds to a cut in $G$. Our formal construction follows.

Let $\tau = 4 \lceil \alpha \rceil |E| + 1$; we subdivide each edge not contained in an edge ring $\tau$ times.
For each face of length $\ell > 3$ we triangulate by adding $\ell + 1$ concentric cycles, each with $\ell$ vertices, labeled $\gamma_0,\dots,\gamma_\ell$, where $\gamma_0$ is the original face from the subdivided version of $\M$. By $v_{i, j}$ we denote the $j$th vertex in $\gamma_i$. We add the edges $(v_{i, j}, v_{i+1, j}$ and $(v_{i, j}, v_{i, j+1 \mod \ell})$.
To complete the triangulation we add one additional vertex $\overline{v}$ at the center of $\gamma_\ell$ and add an edge between it and each vertex on $\gamma_\ell$. We call the new simplicial complex $\M'$. See Figure \ref{fig:face_triangulation} for an example.

\begin{figure}[!htb]
  \centering
    \includegraphics[height=1.5in]{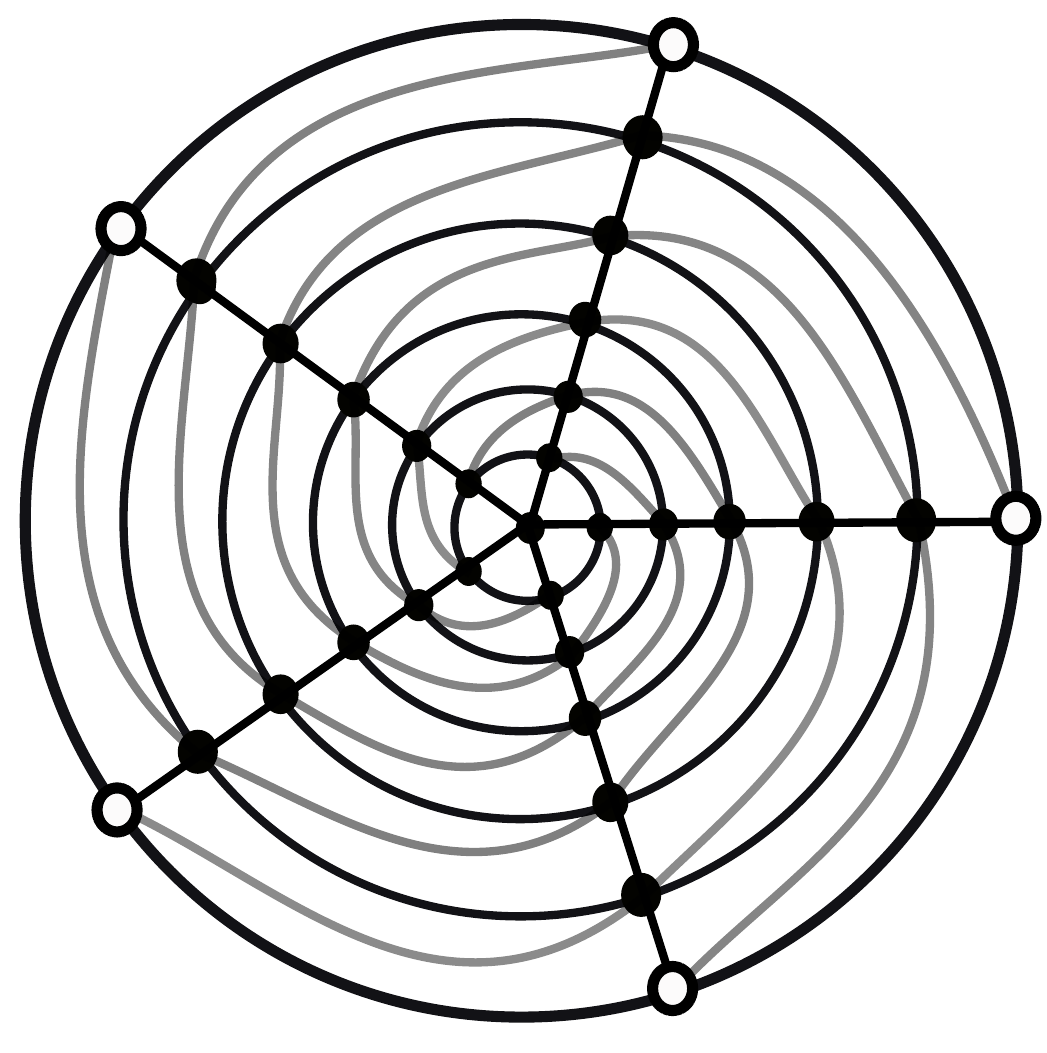}
  \caption{Subdividing a face of length five; the outer face with white vertices is the original face.}
  \label{fig:face_triangulation}
\end{figure}

Let $(S_{opt}, \overline{S_{opt}})$ be an optimal solution to the minimum cut completion instance $(G, E')$.
Suppose we can compute an $\alpha$-approximation $C$ of the minimum homologous cycle instance $(\M', D)$, hence $|C| \leq \alpha |C_{opt}|$. By our construction an optimal solution to $(\M', D)$ has the same size as an optimal solution to $(\M, D)$.
As $C$ is a cycle, if $C$ crosses a cycle $\gamma_0$ it must cross it an even number of times. For any two consecutive vertices $u, v \in \gamma_0$ in $C$ we replace the path between them with the shortest path contained in $\gamma_0$.
We call the new cycle $C'$, since $C' \leq C$ we have that $C'$ is also an $\alpha$-approximation for $(\M', D)$. Note that $C'$ is a union of edge rings, otherwise $|C'| > \alpha |C_{opt}|$.
It follows that $C'$ corresponds to a cut $E_{S'}$ with $|C'| = 4|E_{S'} \oplus E'|$.
Hence, we have $|E_{S'} \oplus E'| \leq \alpha |E_{S_{opt}} \oplus E'|$.
Thus, $E_{S'}$ is an $\alpha$-approximation for $(G, E')$.
\end{proof}
\subsection{Wrap up} 
It remains to show that the cut completion problem is hard to approximate.  We show this via a straightforward reduction from the \emph{minimum uncut problem}: given a graph $G = (V,E)$, find a cut with minimum number of uncut edges.  Note that the optimal cuts for the minimum uncut problem and the maximum cut problem coincide, yet, approximation algorithms for one problem do not necessarily imply approximation algorithm for the other one. 

\begin{lemma}
\label{lem:cut_completion_to_uncut}
The minimum uncut problem is a special case of the minimum cut completion problem.
\end{lemma}
\begin{proof}
Consider the cut completion problem for $G = (V,E)$, and let $E' = E$.  Let $(S, \overline{S})$ be any cut with edge set $E_S$.
The cut completion cost of this cut is 
\[
|E_S \oplus E'| = |E_S \oplus E| = |E\setminus E_S|,
\]
which is the number of uncut edges by $(S, \overline{S})$.
\end{proof}

Now, we are ready to prove our hardness results. 

\begin{proof}[Proof of Theorem~\ref{thm:MBC_hardness} and \ref{thm:MHC_hardness}]
The minimum uncut problem is hard to approximate within $(1+\varepsilon)$ for some $\varepsilon > 0$~\cite{Papadimitriou88OptApxComplexity}.  In addition, it is hard to approximate within any constant factor assuming the unique games conjecture~\cite{Manurangsi:2018, Khot07OptInapprxMaxCut, Chawla06HardnessApxMulticutSparsecut, Khot15UGCIntGapEmb}.   
By Lemma~\ref{lem:cut_completion_to_uncut}, the cut completion problem generalizes the minimum uncut problem. 
Finally, by Lemma~\ref{lem:cut_completion_to_hom_cycle} and \ref{lem:cut_completion_to_bounding_chain_2}, for any $\alpha > 1$ and $\varepsilon > 0$, an $\alpha$-approximation algorithm for the minimum bounded chain problem or the minimum homologous cycle problem implies a $((1+\varepsilon) \alpha)$-approximation algorithm for the cut completion problem.
\end{proof}

\section{A polynomial time special case}\label{sec:exact}
We have shown hardness results, approximation algorithms and parameterized algorithms for the minimum bounded chain problem.
We showed that the problem is hard to approximate even when the input cycle $C$ is on the outer shell of an unweighted $2$-complex embedded in $\R^3$.
If $C$ is null-homologous on the outer shell of $\K$ there is an exact polynomial time algorithm to find the minimum chain bound by $C$.
The assumption that $C$ is null-homologous on the outer shell of $\K$ allows us to treat the problem as a generalization of the shortest $(s, t)$-path problem in planar graphs when $s$ and $t$ are contained on the boundary of the unbounded face.
Hence, we can generalize the duality between shortest paths and minimum cuts in planar graphs to $d$-complexes embedded in $\R^{d+1}$.
The algorithm was first found by Kirsanov and Gortler under the assumption that $H_d(\K)$ is trivial \cite{Kirsanov2004}.
For the sake of completeness we include the same algorithm but described for any $d$-complex embedded in $\R^{d+1}$.
Before describing the algorithm we prove the following lemma about graph cuts, which will be useful in the proof of correctness of the algorithm.

\begin{lemma}
\label{lem:two_cuts}
Let $(S,\overline{S})$ and $(S', \overline{S'})$ be two $(s,t)$-cuts of a graph $G$ with edge sets $E_S$ and $E_{S'}$, respectively. The symmetric difference $E_S\oplus E_{S'}$ is the set of edges of a cut that has $s$ and $t$ on the same side. 
\end{lemma}
\begin{proof}
We show the edge set of the cut $(S\oplus S', \overline{S\oplus S'})$ is $E_S\oplus E_{S'}$.
The statement follows as $s, t \in \overline{S\oplus S'}$.

Let $e = (u,v) \in E_S\oplus E_{S'}$. 
Either, $e\in E_S$ and $e\notin E_{S'}$ or $e\notin E_{S}$ and $e\in E_{S'}$. 

In the first case, there are two possibilities up to symmetry of $(u,v)$.  Either $u\in S\cap S'$ and $v\in \overline{S}\cap S'$, which implies $u\in \overline{S\oplus S'}$ and $v\in S\oplus S'$, or  $u\in S\cap \overline{S'}$ and $v\in \overline{S}\cap \overline{S'}$, which implies $u\in S\oplus S'$ and $v\in \overline{S\oplus S'}$.

In the second case, there are again two possibilities up to symmetry of $(u,v)$. 
Either $u\in S\cap S'$ and $v\in S\cap \overline{S'}$, which implies $u\in \overline{S\oplus S'}$ and $v\in S\oplus S'$, or  $u\in \overline{S}\cap S'$ and $v\in \overline{S}\cap \overline{S'}$, which implies $u\in S\oplus S'$ and $v\in \overline{S\oplus S'}$.
\end{proof}

Let $F\subseteq \shell(\K)$ be a $d$-chain such that $\partial F = C$.  Such an $F$ exists by the assumption of this section. 
We define a cut problem based on $F$. 
Let $\K^*$ be the dual graph of the complex $\K$.
By the definition of $\shell(\K)$, each edge of $(\shell(\K))^*$ is adjacent to $v_\infty$.  
We build the graph $H$ from $\K^*$ by splitting $v_\infty$ as follows.
We replace $v_\infty$ with two vertices $v^+_\infty$ and $v^-_\infty$.  
We replace the incident edges to $v_\infty$ as follows:
\begin{enumerate}
	\item [(i)] A loop that is dual to a face in $F$ is replaced by a $(v^-_\infty, v^+_\infty)$ edge.
	\item [(ii)] A loop that is dual to a face \emph{not} in $F$ is replaced by $v^+_\infty$-loops.
	\item [(iii)] A non-loop edge $(v_\infty, u)$ that is dual to a face in $F$ is replaced by a $(v^-_\infty, u)$-edge.
	\item [(iv)] A non-loop edge $(v_\infty, u)$ that is dual to a face \emph{not} in $F$ is replaced by a $(v^+_\infty, u)$-edge.
\end{enumerate}

\begin{figure}[!htb]
  \centering
    \includegraphics[height=2.0in]{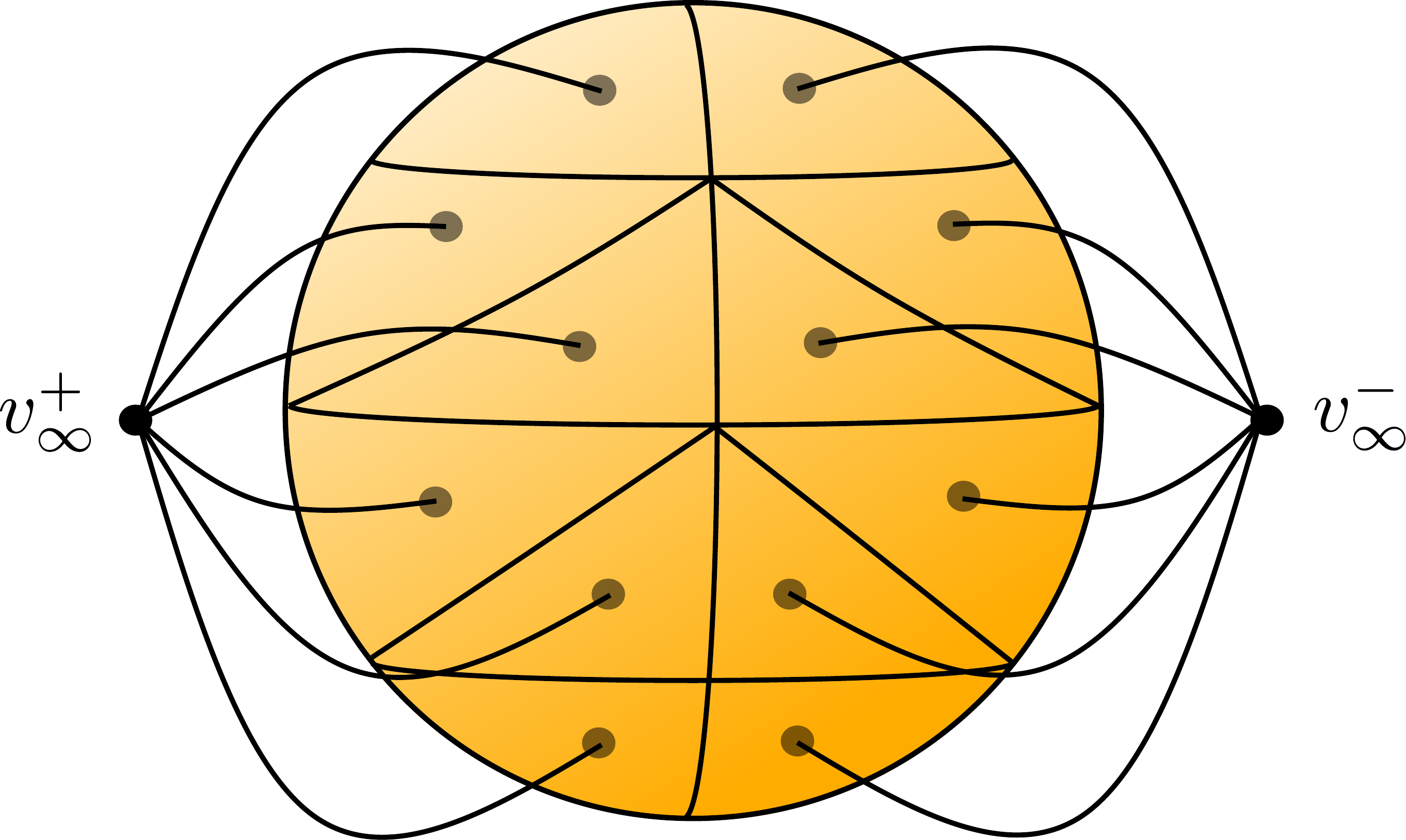}
  \caption{The modified dual graph of a simplicial complex whose outer shell is a triangulated sphere. The vertical line represents the boundary input boundary $C$ which partitions $\shell(\K)$ into two regions.}  
  \label{fig:thickend}
\end{figure}

Note that all of the faces of $F$ correspond to edges that are incident to $v^-_\infty$.
We are now ready to prove the main theorem of the section.

\begin{theorem}\label{2cyclecorrectness}
Let $\K$ be a simplicial complex embedded in $\R^{d+1}$ and $C$ be a null-homologous $(d-1)$-cycle in $\shell(\K)$.
A $d$-chain $D$ is a minimum $d$-chain bounded by $C$ if and only if $D^*$ is a minimum $(v^+_\infty, v^-_\infty)$-cut in $H$.
\end{theorem}
\begin{proof}
We show a one-to-one correspondence between $d$-chains with boundary $C$ in $\K$ and $(v^-_\infty, v^+_\infty)$-cuts in $H$ that preserves the cost.  
Let $D$ be a $d$-chain with $\partial D = C$.  Since $C$ is null-homologous in $\shell(\K)$, there exists $F\subseteq \shell(\K)$ such that $\partial F = C$.
It follows that $\partial(D \oplus F) = \partial D \oplus \partial F = 0$, that is $D\oplus F$ is a $d$-cycle.  
Thus, by cylce/cut duality $D^* \oplus F^*$ is a cut in $\K^*$ that partitions the vertices into two sets $X$ and $Y$. 
Assume, without loss of generality, that $v_\infty \in X$, and note that by splitting $v_\infty$ we obtain a $(X\setminus\{v_\infty\}\cup\{v^-_\infty, v^+_\infty\}, Y)$-cut.  Hence, $D^* \oplus F^*$ is a cut in both $\K^*$ and $H$.

We show that any simple $(v^-_\infty, v^+_\infty)$-path of $H$ crosses $D^*$, therefore $D^*$ is a $(v^-_\infty, v^+_\infty)$-cut.
Let $\beta = (v^-_\infty = v_0, v_1, \ldots, v_k = v^+_\infty)$ be a simple $(v^-_\infty, v^+_\infty)$-path in $H$.
Let $\alpha$ be the closed simple cycle in $\K^*$ obtained by identifying $v_0$ and $v_k$ in $\beta$.
Since $\alpha$ is a closed cycle and $D^* \oplus F^*$ is a cut in $\K^*$, $\alpha$ crosses $D^* \oplus F^*$ an even number of times. 
Therefore, $\beta$ crosses $D^* \oplus F^*$ in $H$ an even number of times; as each edge of $\beta$ is in $D^* \oplus F^*$ in $H$ if and only if the corresponding edge of it in $\alpha$ is in $D^* \oplus F^*$ in $\K^*$.
On the other hand, $v_0 = v^-_\infty$ is only incident to edges from $F^*$.  In particular, $(v_0, v_1) \in F^*$.  If $(v_0, v_1) \in D^*$, then $\beta$ crosses $D^*$ and so the statement holds.  Otherwise, if $(v_0, v_1) \notin D^*$, then the path $(v_1, \ldots, v_k)$ must cross $D^* \oplus F^*$ at least once.  Since all $F$-edges are incident to $v_0$ and $\beta$ is simple we have $v_i \neq v_0$ for any $0<i\leq k$.  Therefore, $(v_1, \ldots, v_k)$ must cross $D^*$ and so the statement holds. 

Conversely, let $D^*$ be a $(v^-_\infty, v^+_\infty)$-cut in $H$.  Since $F^*$ is composed of all edges incident to $v^-_\infty$, it is a $(v^-_\infty, v^+_\infty)$-cut, too. 
It follows that $D^* \oplus F^*$ is a cut in $H$ that has $v^-_\infty$ and $v^+_\infty$ on the same side by Lemma~\ref{lem:two_cuts}.
Therefore, $D^* \oplus F^*$ is a cut in $\K^*$; obtained after identifying $v^-_\infty$ and $v^+_\infty$.  Now, by cycle/cut duality, $D\oplus F$ is a cycle in $\K$, that is $\partial(D\oplus F) = 0$.  As $\partial F = C$, we have $\partial D = C$ and the proof is complete. 
\end{proof}

Now we compute the runtime of the presented algorithm.
The time required to perform the minimum cut computation dominates the preprocessing we perform on the dual graph.
A minimum $(s, t)$-cut in a graph with $n$ vertices and $m$ edges can be computed in $O(nm)$ time via the maximum flow algorithm of Orlin \cite{Orlin2013}.
If $\K$ has $m$ facets then the dual graph $\K^*$ will have $m$ edges.
The number of vertices in $H$ is equal to $\beta_d + 2$. Hence, we can compute the cut in $O(\beta_d m)$ time.

\paragraph{Acknowledgements}
This material is based upon work
supported by the National Science Foundation under Grant Nos.\
CCF-1617951 and CCF-1816442.

\bibliographystyle{plain}
\bibliography{minchain}
\appendix

\end{document}